\documentclass[runningheads]{llncs}
\usepackage{graphicx}
\usepackage[hidelinks,urlcolor=blue]{hyperref}
\hypersetup{colorlinks=true,linkcolor=blue,citecolor=blue,linktocpage}

\usepackage{mathtools}
\usepackage{multido}
\usepackage{xcolor}
\usepackage{tikz}
\usetikzlibrary{calc}
\usetikzlibrary{shapes.multipart}
\usetikzlibrary{shapes.geometric}
\usetikzlibrary{snakes}
\usetikzlibrary{fit}
\usetikzlibrary{arrows.meta}

\usepackage{todonotes,enumitem}

\usepackage{subcaption}
\newsavebox{\imagebox}

\graphicspath{{./figures/}}
\newtheorem{obs}[theorem]{Observation}

\let\oldendproof\endproof
\renewcommand\endproof{~\hfill$\qed$\oldendproof}
  
\begin{document}

\renewcommand{\orcidID}[1]{\href{https://orcid.org/#1}{\includegraphics[scale=.03]{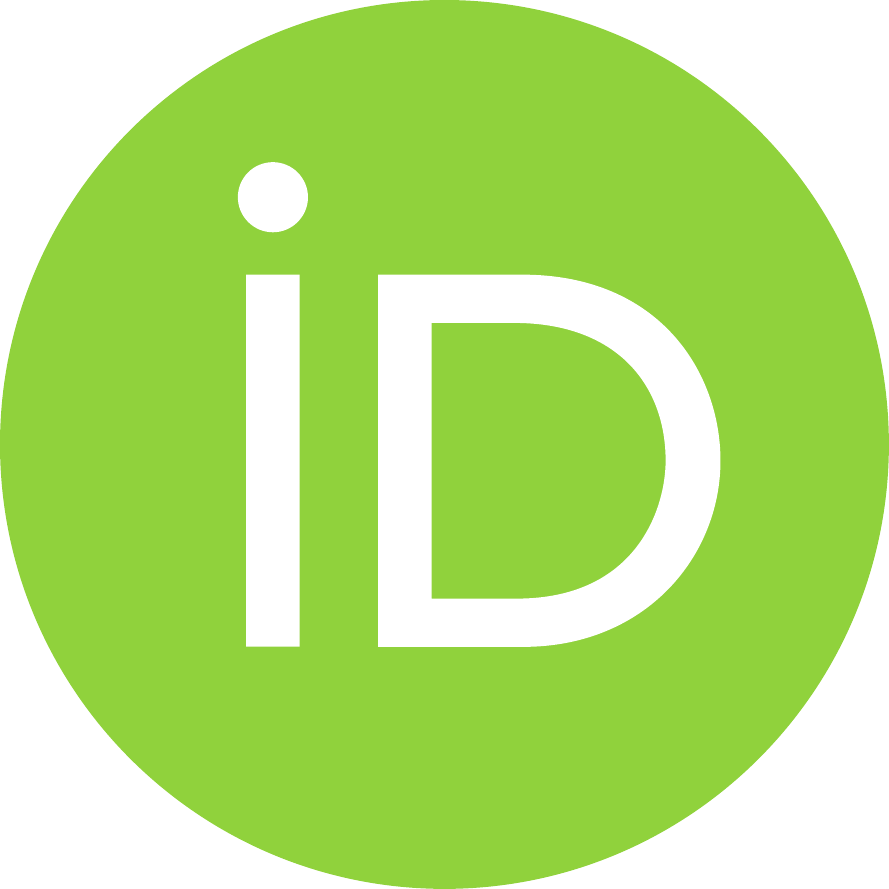}}}

\title{On the Maximum Number of Crossings in Star-Simple Drawings of $K_n$\\ with No Empty Lens%
\thanks{This research started at the 3rd
    Workshop within the collaborative DACH project \emph{Arrangements
        and Drawings}, August 19--23, 2019, in Wergenstein (GR),
    Switzerland, supported by the German Research Foundation (DFG),
    the Austrian Science Fund (FWF), and the Swiss National Science
    Foundation (SNSF). We thank the participants for stimulating
    discussions. S.F.~is supported by DFG Project FE~340/12-1. M.H.~is
    supported by SNSF Project 200021E-171681. K.K.~is supported by DFG
    Project MU~3501/3-1 and within the Research Training Group
    GRK~2434 \emph{Facets of Complexity}. I.P.~was supported
    by FWF project I~3340-N35.}}
\titlerunning{Crossings in Star-Simple Drawings of $K_n$ with No Empty Lens}
%
\author{Stefan Felsner\inst{1}\orcidID{0000-0002-6150-1998} 
\and
Michael Hoffmann\inst{2}\orcidID{0000-0001-5307-7106} 
\and
Kristin Knorr\inst{3}\orcidID{0000-0003-4239-424X}
\and
Irene Parada\inst{4}\orcidID{0000-0003-3147-0083}
}
\authorrunning{S. Felsner et al.}
%
\institute{Institute of Mathematics, Technische Universität Berlin, Germany\\
    \email{felsner@math.tu-berlin.de}\and
    Department of Computer Science, ETH Zürich, Switzerland\\
    \email{hoffmann@inf.ethz.ch}\and
    Department of Computer Science, Freie Universität Berlin, Germany\\
    \email{knorrkri@inf.fu-berlin.de} \and
    Department of Mathematics and Computer Science, \\
    TU Eindhoven, The Netherlands \email{i.m.de.parada.munoz@tue.nl}}
\maketitle              
\begin{abstract}
    A star-simple drawing of a graph is a drawing in which adjacent
    edges do not cross. 
    In contrast,
    there is no restriction on the number of crossings between two
    independent edges. When allowing empty lenses (a face in the
    arrangement induced by two edges that is bounded by a $2$-cycle),
    two independent edges may cross arbitrarily many times in a
    star-simple drawing. We consider star-simple drawings of~$K_n$
    with no empty lens. In this setting we prove an upper bound
    of~$3((n-4)!)$ on the maximum number of crossings between any pair
    of edges. It follows that the total number of crossings is finite
    and upper bounded by~$n!$.
    
\end{abstract}

\section{Introduction}
A \emph{topological drawing} of a graph~$G$ is a drawing in the plane
where vertices are represented by pairwise distinct points, and edges
are represented by Jordan arcs with their vertices as endpoints.
Additionally, edges do not contain any other vertices, every common
point of two edges is either a proper crossing or a common endpoint,
and no three edges cross at a single point.  A \emph{simple drawing}
is a topological drawing in which adjacent edges do not cross, and
independent edges cross at most once.

We study a broader class of topological drawings, which are called \emph{star-simple} drawings, where
adjacent edges do not cross, but independent edges may cross any
number of times; see \figurename~\ref{fig:simple:ex} for
illustration. In such a drawing, for every vertex~$v$ the induced
substar centered at~$v$ is simple, that is, the drawing restricted to
the edges incident to~$v$ forms a plane drawing. In the literature
(e.g.,~\cite{semisimple17,balko2015crossing}) these drawings also
appear under the name \emph{semi-simple}, but we prefer star-simple
because the name is much more descriptive.

\begin{figure}[htb]
    \begin{minipage}[t]{.29\linewidth}
        \centering\includegraphics[page=4]{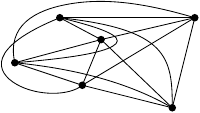}
        \subcaption{simple}\label{fig:simple:ex:1}
    \end{minipage}\hfill
    \begin{minipage}[t]{.35\linewidth}
        \centering\includegraphics[page=6]{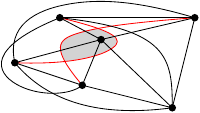}
        \subcaption{star-simple but not simple}\label{fig:simple:ex:2}
    \end{minipage}\hfill
    \begin{minipage}[t]{.35\linewidth}
        \centering\includegraphics[page=5]{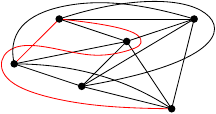} 
        \subcaption{not star-simple}\label{fig:simple:ex:3}
    \end{minipage}
    \caption{Topological drawings of $K_6$ and a (nonempty) lens (shaded in (b)).\label{fig:simple:ex}}
\end{figure}

In contrast to simple drawings, star-simple drawings can have regions or cells whose boundary consists
of two continuous pieces of (two) edges. We call such a region a
\emph{lens}; see \figurename~\ref{fig:simple:ex:2}. A lens is \emph{empty} if it has no vertex in its
interior.  If empty lenses are
allowed, the number of crossings in star-simple drawings of graphs with
at least two edges is unbounded (twisting), as illustrated in \figurename~\ref{fig:twisting}.
We restrict our attention
to star-simple drawings with no empty lens. This restriction is---in
general---not sufficient to guarantee a bounded number of crossings
(spiraling), as illustrated in \figurename~\ref{fig:spiraling}. However, we will show that star-simple drawings \emph{of
    the complete graph~$K_n$} with no empty lens have a bounded number
of crossings.

\begin{figure}[bht]%
  \captionsetup[subfigure]{justification=centering}
    \hbox{}\hfill\begin{minipage}[b]{.4\linewidth}
        \centering
        \begin{tikzpicture}[scale = 0.8]
	\coordinate (S) at (0,0);

	
	\draw[thick] ($(S)+(0.5,0.3)$) -- ($(S)+(0.5,0)$);
	\draw[thick] ($(S)+(0.5,0)$) arc (180:360:0.3);
	\draw[thick] ($(S)+(1.1,0)$) arc (180:0:0.3);
	\draw[thick] ($(S)+(1.7,0)$) arc (180:360:0.3);
	\draw[thick] ($(S)+(2.3,0)$) arc (180:0:0.3);
	\draw[thick] ($(S)+(2.9,0)$) arc (180:360:0.3);
	\draw[thick] ($(S)+(3.5,0)$) arc (180:0:0.3);
  \draw[thick] ($(S)+(4.1,-0.3)$) -- ($(S)+(4.1,0)$);
	
	\draw[color=red,thick] ($(S)+(0,0)$) -- ($(S)+(4.6,0)$);
		
	\fill ($(S)+(0.5,0.3)$) circle (0.056);
	\fill ($(S)+(4.1,-0.3)$) circle (0.056);
	\fill (S) circle (0.056);
	\fill ($(S)+(4.6,0)$) circle (0.056);
	
\end{tikzpicture}\vspace{38pt}
        \subcaption{twisting}
        \label{fig:twisting}
    \end{minipage}\hfill
    \begin{minipage}[b]{.4\linewidth}
        \centering
        \begin{tikzpicture}[scale=0.8]
	
	\coordinate (S1) at (0,0);
	\draw[thick] ($(S1)+(0.5,0.5)$) -- ($(S1)+(0.5,0)$);
	\draw[thick] ($(S1)+(0.5,0)$) arc (180:360:2);
	\draw[thick] ($(S1)+(4.5,0)$) arc (0:180:1.75);
	\draw[thick] ($(S1)+(1,0)$) arc (180:360:1.5);
	\draw[thick] ($(S1)+(4,0)$) arc (0:180:1.25);
	\draw[thick] ($(S1)+(1.5,0)$) arc (180:360:1);
	\draw[thick] ($(S1)+(3.5,0)$) arc (0:180:0.75);
	\draw[thick] ($(S1)+(2,0)$) arc (180:360:0.5);
	\draw[thick] ($(S1)+(3,0)$) arc (0:180:0.25);
	\draw[thick] ($(S1)+(2.5,-0.25)$) -- ($(S1)+(2.5,0)$);

	\draw[color=red,thick] ($(S1)+(0,0)$) -- ($(S1)+(5,0)$);
		
	\fill ($(S1)+(0.5,0.5)$) circle (0.056);
	\fill ($(S1)+(2.5,-0.25)$) circle (0.056);
	\fill ($(S1)+(2.75,0.125)$) circle (0.056);
	\fill (S1) circle (0.056);
	\fill ($(S1)+(5,0)$) circle (0.056);
	
\end{tikzpicture}
        \subcaption{spiraling}
        \label{fig:spiraling}
    \end{minipage}\hfill\hbox{}
    \caption{Constructions to achieve an unbounded number of crossings.}
    \label{fig:unb_crossings}
\end{figure}
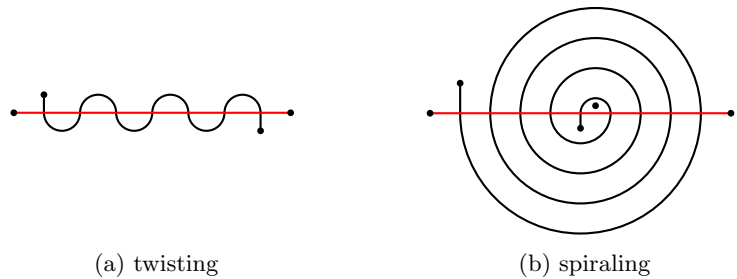

Empty lenses also play a role in the context of the crossing lemma for
multigraphs~\cite{PachToth2020}. This is because a group of
arbitrarily many parallel edges can be drawn without a single
crossing.  Hence, for general multigraphs there is no hope to get a
lower bound on the number of crossings as a function of the number of
edges.  However, if we forbid empty lenses, we cannot draw arbitrarily many parallel edges.

Kyn\v{c}l~\cite[Section 5 "Picture hanging without
crossings"]{Kyncl2016} proposed a construction of two edges in a graph
on~$n$ vertices with an exponential number~($2^{n-4}$) of crossings
and no empty lens; see \figurename~\ref{fig:exp_crossings}.  This
configuration can be completed to a star-simple drawing of~$K_n$, cf.~\cite{parada_thesis}. 
For $n=6$ it is possible to have one more crossing while maintaining
the property that the drawing can be completed to a star-simple
drawing of~$K_6$; see \figurename~\ref{fig:many-crossings}. Repeated
application of the doubling construction of
\figurename~\ref{fig:exp_crossings} leads to two edges with~$2^{n-4}+2^{n-6}$ crossings in a graph on~$n$ vertices.  This
configuration can be completed to a star-simple drawing of~$K_n$. We
suspect that this is the maximum number of crossings of two edges in a
star-simple drawing of~$K_n$.

\begin{figure}[htb]%
    \centering
    \begin{tikzpicture}[scale=0.67]
	\coordinate (S) at (0,0);
	 \draw[thick] ($(S)+(1,1)$) -- ($(S)+(1,-1)$);
  \fill ($(S)+(1,1)$) circle (0.07);
  \fill ($(S)+(1,-1)$) circle (0.07);

  \draw[color=red,thick] (S) -- ($(S)+(2,0)$);
  \fill (S) circle (0.07);
  \fill ($(S)+(2,0)$) circle (0.07);
  
  \draw[|-{>[width=3mm]},thick] ($(S)+(2.4,0)$) -- ($(S)+(3.1,0)$);
  \coordinate (S1) at ($(S)+(3.5,0)$);
	\draw[thick] ($(S1)+(0.5,0)$) -- ($(S1)+(0.5,-1)$);
	\draw[thick] ($(S1)+(2.5,0)$) -- ($(S1)+(2.5,-1)$);
	\draw[thick] ($(S1)+(0.5,0)$) arc (180:0:1);
  \fill ($(S1)+(0.5,-1)$)circle (0.07);
  \fill ($(S1)+(2.5,-1)$) circle (0.07);

  \draw[color=red,thick] (S1) -- ($(S1)+(3,0)$);
  \fill (S1) circle (0.07);
  \fill ($(S1)+(3,0)$) circle (0.07);
  
  \fill ($(S1)+(1.5,0.5)$) circle (0.07);
  
  \draw[|-{>[width=3mm]},thick] ($(S)+(6.9,0)$) -- ($(S)+(7.6,0)$);
  \coordinate (S2) at ($(S)+(8,0)$);
	\draw[thick] ($(S2)+(0.5,0)$) -- ($(S2)+(0.5,-0.75)$);
	\draw[thick] ($(S2)+(1.5,0)$) -- ($(S2)+(1.5,-0.75)$);
	\draw[thick] ($(S2)+(0.5,0)$) arc (180:0:1.5);
	\draw[thick] ($(S2)+(1.5,0)$) arc (180:0:0.5);
	\draw[thick] ($(S2)+(0.5,0)$) arc (180:0:1.5);
	\draw[thick] ($(S2)+(1.5,0)$) arc (180:0:0.5);
	\draw[thick] ($(S2)+(2.5,0)$) arc (180:360:0.5);
  \fill ($(S2)+(0.5,-0.75)$)circle (0.07);
  \fill ($(S2)+(1.5,-0.75)$) circle (0.07);

  \draw[color=red,thick] (S2) -- ($(S2)+(4,0)$);
  \fill (S2) circle (0.07);
  \fill ($(S2)+(4,0)$) circle (0.07);
  
  \fill ($(S2)+(2,0.25)$) circle (0.07);
  \fill ($(S2)+(3,-0.25)$) circle (0.07);
  
  \draw[|-{>[width=3mm]},thick] ($(S)+(12.4,0)$) -- ($(S)+(13.1,0)$);
  \coordinate (S3) at ($(S)+(13.5,0)$);
	\draw[thick] ($(S3)+(0.25,0)$) -- ($(S3)+(0.25,-0.5)$);
	\draw[thick] ($(S3)+(0.75,0)$) -- ($(S3)+(0.75,-0.5)$);
	\draw[thick] ($(S3)+(0.25,0)$) arc (180:0:1.75);
	\draw[thick] ($(S3)+(0.75,0)$) arc (180:0:1.25);
	\draw[thick] ($(S3)+(3.75,0)$) arc (360:180:0.75);
	\draw[thick] ($(S3)+(3.25,0)$) arc (360:180:0.25);
	\draw[thick] ($(S3)+(2.75,0)$) arc (0:180:0.75);
	\draw[thick] ($(S3)+(2.25,0)$) arc (0:180:0.25);
	\draw[thick] ($(S3)+(1.25,0)$) arc (180:360:0.25);
  \fill ($(S3)+(0.25,-0.5)$)circle (0.07);
  \fill ($(S3)+(0.75,-0.5)$) circle (0.07);

  \draw[color=red,thick] (S3) -- ($(S3)+(4,0)$);
  \fill (S3) circle (0.07);
  \fill ($(S3)+(4,0)$) circle (0.07);
  
  \fill ($(S3)+(2,0.125)$) circle (0.07);
  \fill ($(S3)+(1.5,-0.125)$) circle (0.07);
  \fill ($(S3)+(3,-0.125)$) circle (0.07);

  \coordinate (S4) at ($(S)+(2,2.5)$);
  \draw[dashed,color=gray!60!white] (S4) rectangle ($(S4)+(13.5,1)$);
  \draw[|-{>[width=3mm]},thick] ($(S)+(8.4,3)$) -- ($(S)+(9.1,3)$);
  \draw[thick] ($(S4)+(0.25,0.5)$) -- ($(S4)+(6,0.5)$);
  \draw[thick] ($(S4)+(7.5,0.25)$) -- ($(S4)+(13,0.25)$);
  \draw[thick] ($(S4)+(7.5,0.75)$) -- ($(S4)+(13,0.75)$);
  \draw[thick] ($(S4)+(13,0.75)$) arc (270:90:-0.25);
  \fill ($(S4)+(0.25,0.5)$) circle (0.07);
  \fill ($(S4)+(6,0.5)$) circle (0.07);
  \fill ($(S4)+(7.5,0.25)$) circle (0.07);
	\fill ($(S4)+(7.5,0.75)$) circle (0.07);
  \fill ($(S4)+(13,0.5)$) circle (0.07);
  
\end{tikzpicture}
    \caption{The doubling construction yields an exponential number of crossings.}
    \label{fig:exp_crossings}
\end{figure}
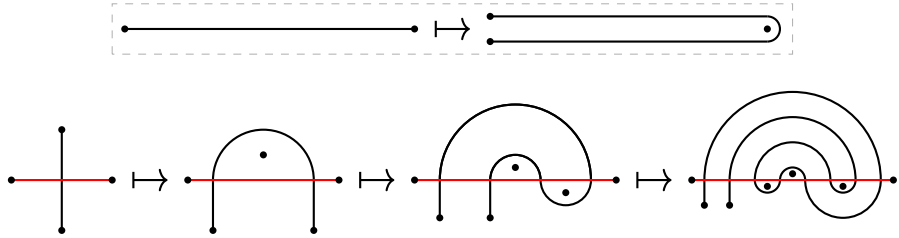

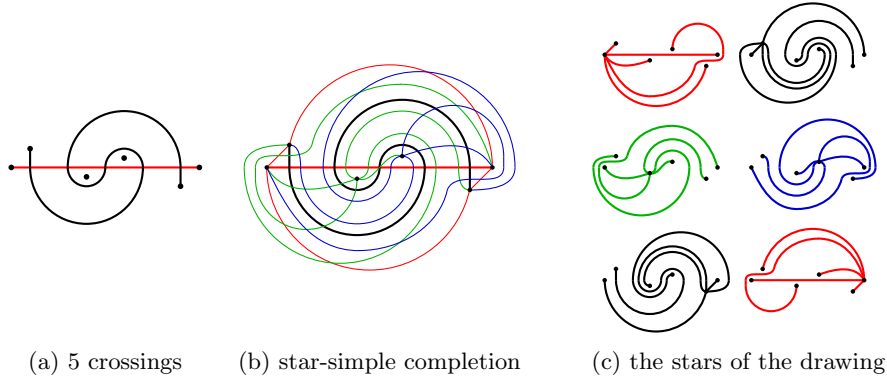
\begin{figure}[bht]
	\centering
	\savebox{\imagebox}{\begin{tikzpicture}[scale=0.15]
	\coordinate (SC) at (0,0);
	
	\coordinate (SCa) at (SC);
	
	\draw[color=red,thick] ($(SCa)+(0,0)$) -- ($(SCa)+(10,0)$);
	\draw[color=red,thick] ($(SCa)+(0,0)$) -- ($(SCa)+(1,1)$);
	\draw[color=red,thick] ($(SCa)+(0,0)$) arc (180:348:4.55);
	\draw[color=red,thick] ($(SCa)+(4,-0.5)$) to [in=300, out=200] (SCa);
	\draw[color=red,thick] ($(SCa)+(6,0.5)$) arc (180:0:2.25);
	\draw[color=red,thick] ($(SCa)+(10.5,0.5)$) -- ($(SCa)+(10.5,0)$);
	\draw[color=red,thick] ($(SCa)+(10.5,0)$) to [in=70, out=270] ($(SCa)+(8,-1)$);
	\draw[color=red,thick] ($(SCa)+(8,-1)$) arc (350:270:4 and 3.25);
	\draw[color=red,thick] ($(SCa)+(4.1,-3.69)$) to [in=295, out=180] (SCa);
		
	\fill ($(SCa)+(1,1)$) circle (0.2);
	\fill ($(SCa)+(9,-1)$) circle (0.2);
	\fill ($(SCa)+(4,-0.5)$) circle (0.2);
	\fill ($(SCa)+(6,0.5)$) circle (0.2);
	\fill (SCa)circle (0.2);
	\fill ($(SCa)+(10,0)$) circle (0.2);
	
	\coordinate (SCb) at ($(SCa)+(13,0)$);
	
	\draw[thick] ($(SCb)+(0,0)$) -- ($(SCb)+(1,1)$);
	\draw[thick] ($(SCb)+(5,0)$) arc (360:180:1);
	\draw[thick] ($(SCb)+(7,0)$) arc (360:180:3);
	\draw[thick] ($(SCb)+(5,0)$) arc (180:0:1);
	\draw[thick] ($(SCb)+(3,0)$) arc (180:0:3);
	\draw[thick] ($(SCb)+(1,0)$) -- ($(SCb)+(1,1)$);
	\draw[thick] ($(SCb)+(9,0)$) -- ($(SCb)+(9,-1)$);
	\draw[thick] ($(SCb)+(10,0)$) arc (0:168:4.55);
	\draw[thick] ($(SCb)+(4.5,0)$) arc (360:270:0.5);
	\draw[thick] ($(SCb)+(4.5,0)$) arc (180:0:1.5);
	\draw[thick] ($(SCb)+(7.5,0)$) arc (360:180:4.25);
	\draw[thick] ($(SCb)+(-1,0)$) to [in=180, out=90] ($(SCb)+(1,1)$);
	\draw[thick] ($(SCb)+(6.5,0)$) arc (0:90:0.5);
	\draw[thick] ($(SCb)+(6.5,0)$) arc (360:180:2.5);
	\draw[thick] ($(SCb)+(1.5,0)$) to [in=0, out=90] ($(SCb)+(1,1)$);
	
	\fill ($(SCb)+(1,1)$) circle (0.2);
	\fill ($(SCb)+(9,-1)$) circle (0.2);
	\fill ($(SCb)+(4,-0.5)$) circle (0.2);
	\fill ($(SCb)+(6,0.5)$) circle (0.2);
	\fill (SCb)circle (0.2);
	\fill ($(SCb)+(10,0)$) circle (0.2);
	
	\coordinate (SCc) at ($(SCa)+(0,-10)$);
	
	\draw[color=green!70!black,thick] ($(SCc)+(4,-0.5)$) to [in=300, out=200] (SCc);
	\draw[color=green!70!black,thick] ($(SCc)+(4.5,0)$) arc (360:270:0.5);
	\draw[color=green!70!black,thick] ($(SCc)+(4.5,0)$) arc (180:0:1.5);
	\draw[color=green!70!black,thick] ($(SCc)+(7.5,0)$) arc (360:180:4.25);
	\draw[color=green!70!black,thick] ($(SCc)+(-1,0)$) to [in=180, out=90] ($(SCc)+(1,1)$);
	\draw[color=green!70!black,thick] ($(SCc)+(4,-0.5)$) to [in=160, out=340] ($(SCc)+(6,0.5)$);
	\draw[color=green!70!black,thick] ($(SCc)+(3.5,0)$) arc (180:270:0.5);
	\draw[color=green!70!black,thick] ($(SCc)+(3.5,0)$) arc (180:0:2.5);
	\draw[color=green!70!black,thick] ($(SCc)+(8.5,0)$) to [in=180, out=270] ($(SCc)+(9,-1)$);
	\draw[color=green!70!black,thick] ($(SCc)+(4,-0.5)$) arc (360:180:2.25);
	\draw[color=green!70!black,thick] ($(SCc)+(-0.5,-0.5)$) -- ($(SCc)+(-0.5,0)$);
	\draw[color=green!70!black,thick] ($(SCc)+(-0.5,0)$) to [in=250, out=90] ($(SCc)+(2,1)$);
	\draw[color=green!70!black,thick] ($(SCc)+(2,1)$) arc (170:90:4 and 3.25);
	\draw[color=green!70!black,thick] ($(SCc)+(5.9,3.69)$) to [in=105, out=0] ($(SCc)+(10,0)$);
	
	\fill ($(SCc)+(1,1)$) circle (0.2);
	\fill ($(SCc)+(9,-1)$) circle (0.2);
	\fill ($(SCc)+(4,-0.5)$) circle (0.2);
	\fill ($(SCc)+(6,0.5)$) circle (0.2);
	\fill (SCc)circle (0.2);
	\fill ($(SCc)+(10,0)$) circle (0.2);
	
	\coordinate (SCd) at ($(SCa)+(13,-10)$);
	
	\draw[color=blue!80!black,thick] ($(SCd)+(6,0.5)$) arc (180:0:2.25);
	\draw[color=blue!80!black,thick] ($(SCd)+(10.5,0.5)$) -- ($(SCd)+(10.5,0)$);
	\draw[color=blue!80!black,thick] ($(SCd)+(10.5,0)$) to [in=70, out=270] ($(SCd)+(8,-1)$);
	\draw[color=blue!80!black,thick] ($(SCd)+(8,-1)$) arc (350:270:4 and 3.25);
	\draw[color=blue!80!black,thick] ($(SCd)+(4.1,-3.69)$) to [in=295, out=180] (SCd);
	\draw[color=blue!80!black,thick] ($(SCd)+(6.5,0)$) arc (0:90:0.5);
	\draw[color=blue!80!black,thick] ($(SCd)+(6.5,0)$) arc (360:180:2.5);
	\draw[color=blue!80!black,thick] ($(SCd)+(1.5,0)$) to [in=0, out=90] ($(SCd)+(1,1)$);
	\draw[color=blue!80!black,thick] ($(SCd)+(4,-0.5)$) to [in=160, out=340] ($(SCd)+(6,0.5)$);
	\draw[color=blue!80!black,thick] ($(SCd)+(6,0.5)$) to [in=120, out=20] ($(SCd)+(10,0)$) ;
	\draw[color=blue!80!black,thick] ($(SCd)+(5.5,0)$) arc (180:90:0.5);
	\draw[color=blue!80!black,thick] ($(SCd)+(5.5,0)$) arc (360:180:1.5);
	\draw[color=blue!80!black,thick] ($(SCd)+(2.5,0)$) arc (180:0:4.25);
	\draw[color=blue!80!black,thick] ($(SCd)+(11,0)$) to [in=0, out=270] ($(SCd)+(9,-1)$);
	
	\fill ($(SCd)+(1,1)$) circle (0.2);
	\fill ($(SCd)+(9,-1)$) circle (0.2);
	\fill ($(SCd)+(4,-0.5)$) circle (0.2);
	\fill ($(SCd)+(6,0.5)$) circle (0.2);
	\fill (SCd)circle (0.2);
	\fill ($(SCd)+(10,0)$) circle (0.2);
	
	\coordinate (SCe) at ($(SCa)+(0,-20)$);
	
	\draw[thick] ($(SCe)+(0,0)$) arc (180:348:4.55);
	\draw[thick] ($(SCe)+(5,0)$) arc (360:180:1);
	\draw[thick] ($(SCe)+(7,0)$) arc (360:180:3);
	\draw[thick] ($(SCe)+(5,0)$) arc (180:0:1);
	\draw[thick] ($(SCe)+(3,0)$) arc (180:0:3);
	\draw[thick] ($(SCe)+(1,0)$) -- ($(SCe)+(1,1)$);
	\draw[thick] ($(SCe)+(9,0)$) -- ($(SCe)+(9,-1)$);
	\draw[thick] ($(SCe)+(3.5,0)$) arc (180:270:0.5);
	\draw[thick] ($(SCe)+(3.5,0)$) arc (180:0:2.5);
	\draw[thick] ($(SCe)+(8.5,0)$) to [in=180, out=270] ($(SCe)+(9,-1)$);
	\draw[thick] ($(SCe)+(5.5,0)$) arc (180:90:0.5);
	\draw[thick] ($(SCe)+(5.5,0)$) arc (360:180:1.5);
	\draw[thick] ($(SCe)+(2.5,0)$) arc (180:0:4.25);
	\draw[thick] ($(SCe)+(11,0)$) to [in=0, out=270] ($(SCe)+(9,-1)$);
	\draw[thick] ($(SCe)+(10,0)$) -- ($(SCe)+(9,-1)$);
	
	\fill ($(SCe)+(1,1)$) circle (0.2);
	\fill ($(SCe)+(9,-1)$) circle (0.2);
	\fill ($(SCe)+(4,-0.5)$) circle (0.2);
	\fill ($(SCe)+(6,0.5)$) circle (0.2);
	\fill (SCe)circle (0.2);
	\fill ($(SCe)+(10,0)$) circle (0.2);
	
	\coordinate (SCf) at ($(SCa)+(13,-20)$);
	
	\draw[color=red,thick] ($(SCf)+(0,0)$) -- ($(SCf)+(10,0)$);
	\draw[color=red,thick] ($(SCf)+(10,0)$) arc (0:168:4.55);
	\draw[color=red,thick] ($(SCf)+(4,-0.5)$) arc (360:180:2.25);
	\draw[color=red,thick] ($(SCf)+(-0.5,-0.5)$) -- ($(SCf)+(-0.5,0)$);
	\draw[color=red,thick] ($(SCf)+(-0.5,0)$) to [in=250, out=90] ($(SCf)+(2,1)$);
	\draw[color=red,thick] ($(SCf)+(2,1)$) arc (170:90:4 and 3.25);
	\draw[color=red,thick] ($(SCf)+(5.9,3.69)$) to [in=105, out=0] ($(SCf)+(10,0)$);
	\draw[color=red,thick] ($(SCf)+(6,0.5)$) to [in=120, out=20] ($(SCf)+(10,0)$) ;
	\draw[color=red,thick] ($(SCf)+(10,0)$) -- ($(SCf)+(9,-1)$);
	
	\fill ($(SCf)+(1,1)$) circle (0.2);
	\fill ($(SCf)+(9,-1)$) circle (0.2);
	\fill ($(SCf)+(4,-0.5)$) circle (0.2);
	\fill ($(SCf)+(6,0.5)$) circle (0.2);
	\fill (SCf)circle (0.2);
	\fill ($(SCf)+(10,0)$) circle (0.2);

\end{tikzpicture}}
    \begin{minipage}[t]{.25\linewidth}
        \centering
        \raisebox{\dimexpr.5\ht\imagebox - .5\height}{
        \begin{tikzpicture}[scale=0.25]
	\coordinate (S) at (0,0);
	
	\draw[color=red,thick] ($(S)+(0,0)$) -- ($(S)+(10,0)$);
	
	\draw[thick] ($(S)+(5,0)$) arc (360:180:1);
	\draw[thick] ($(S)+(7,0)$) arc (360:180:3);
	\draw[thick] ($(S)+(5,0)$) arc (180:0:1);
	\draw[thick] ($(S)+(3,0)$) arc (180:0:3);
	\draw[thick] ($(S)+(1,0)$) -- ($(S)+(1,1)$);
	\draw[thick] ($(S)+(9,0)$) -- ($(S)+(9,-1)$);
	
	\fill ($(S)+(1,1)$) circle (0.15);
	\fill ($(S)+(9,-1)$) circle (0.15);
	\fill ($(S)+(4,-0.5)$) circle (0.15);
	\fill ($(S)+(6,0.5)$) circle (0.15);
	\fill (S) circle (0.15);
	\fill ($(S)+(10,0)$) circle (0.15);

\end{tikzpicture}}
        \subcaption{5~crossings}\label{fig:mc0}
    \end{minipage}
    \begin{minipage}[t]{.33\linewidth}
        \centering
        \raisebox{\dimexpr.5\ht\imagebox - .5\height}{
        \begin{tikzpicture}[scale=0.3]
	\coordinate (S1) at (0,0);
	

	\draw[color=red,thick] ($(S1)+(0,0)$) -- ($(S1)+(10,0)$);
	
	\draw[thick] ($(S1)+(5,0)$) arc (360:180:1);
	\draw[thick] ($(S1)+(7,0)$) arc (360:180:3);
	\draw[thick] ($(S1)+(5,0)$) arc (180:0:1);
	\draw[thick] ($(S1)+(3,0)$) arc (180:0:3);
	\draw[thick] ($(S1)+(1,0)$) -- ($(S1)+(1,1)$);
	\draw[thick] ($(S1)+(9,0)$) -- ($(S1)+(9,-1)$);
	
	\draw[color=red] ($(S1)+(0,0)$) -- ($(S1)+(1,1)$);
	\draw[color=red] ($(S1)+(10,0)$) -- ($(S1)+(9,-1)$);
	\draw[color=red] ($(S1)+(0,0)$) arc (180:348:4.55);
	\draw[color=red] ($(S1)+(10,0)$) arc (0:168:4.55);
	
	\draw[color=green!70!black] ($(S1)+(4,-0.5)$) to [in=160, out=340] ($(S1)+(6,0.5)$);
	\draw[color=green!70!black] ($(S1)+(4,-0.5)$) to [in=300, out=200] (S1);
	\draw[color=green!70!black] ($(S1)+(4.5,0)$) arc (360:270:0.5);
	\draw[color=green!70!black] ($(S1)+(4.5,0)$) arc (180:0:1.5);
	\draw[color=green!70!black] ($(S1)+(7.5,0)$) arc (360:180:4.25);
	\draw[color=green!70!black] ($(S1)+(-1,0)$) to [in=180, out=90] ($(S1)+(1,1)$);
	\draw[color=green!70!black] ($(S1)+(3.5,0)$) arc (180:270:0.5);
	\draw[color=green!70!black] ($(S1)+(3.5,0)$) arc (180:0:2.5);
	\draw[color=green!70!black] ($(S1)+(8.5,0)$) to [in=180, out=270] ($(S1)+(9,-1)$);
	
	\draw[color=green!70!black] ($(S1)+(4,-0.5)$) arc (360:180:2.25);
	\draw[color=green!70!black] ($(S1)+(-0.5,-0.5)$) -- ($(S1)+(-0.5,0)$);
	\draw[color=green!70!black] ($(S1)+(-0.5,0)$) to [in=250, out=90] ($(S1)+(2,1)$);
	\draw[color=green!70!black] ($(S1)+(2,1)$) arc (170:90:4 and 3.25);
	\draw[color=green!70!black] ($(S1)+(5.9,3.69)$) to [in=105, out=0] ($(S1)+(10,0)$);
	
	\draw[color=blue!80!black] ($(S1)+(6,0.5)$) to [in=120, out=20] ($(S1)+(10,0)$) ;
	\draw[color=blue!80!black] ($(S1)+(5.5,0)$) arc (180:90:0.5);
	\draw[color=blue!80!black] ($(S1)+(5.5,0)$) arc (360:180:1.5);
	\draw[color=blue!80!black] ($(S1)+(2.5,0)$) arc (180:0:4.25);
	\draw[color=blue!80!black] ($(S1)+(11,0)$) to [in=0, out=270] ($(S1)+(9,-1)$);
	\draw[color=blue!80!black] ($(S1)+(6.5,0)$) arc (0:90:0.5);
	\draw[color=blue!80!black] ($(S1)+(6.5,0)$) arc (360:180:2.5);
	\draw[color=blue!80!black] ($(S1)+(1.5,0)$) to [in=0, out=90] ($(S1)+(1,1)$);
	\draw[color=blue!80!black] ($(S1)+(6,0.5)$) arc (180:0:2.25);
	\draw[color=blue!80!black] ($(S1)+(10.5,0.5)$) -- ($(S1)+(10.5,0)$);
	\draw[color=blue!80!black] ($(S1)+(10.5,0)$) to [in=70, out=270] ($(S1)+(8,-1)$);
	\draw[color=blue!80!black] ($(S1)+(8,-1)$) arc (350:270:4 and 3.25);
	\draw[color=blue!80!black] ($(S1)+(4.1,-3.69)$) to [in=295, out=180] (S1);

	\fill ($(S1)+(1,1)$) circle (0.1);
	\fill ($(S1)+(9,-1)$) circle (0.1);
	\fill ($(S1)+(4,-0.5)$) circle (0.1);
	\fill ($(S1)+(6,0.5)$) circle (0.1);
	\fill (S1) circle (0.1);
	\fill ($(S1)+(10,0)$) circle (0.1);

\end{tikzpicture}}
        \subcaption{star-simple completion
        }\label{fig:mc-all}
    \end{minipage}\hfill
    \begin{minipage}[t]{.37\linewidth}
        \centering
        \usebox{\imagebox}
        \subcaption{the stars of the drawing
        }\label{fig:mc-stars}
    \end{minipage}
    \caption{Two edges with $2^{n-4}+2^{n-6}$ crossings in a
        star-simple drawing of $K_n$, for $n=6$.\label{fig:many-crossings}} 
\end{figure}

\section{Crossing patterns}

In this section we study the induced drawing~$D(e,e')$ of two independent
edges~$e$ and $e'$ in a star-simple drawing~$D$ of the complete
graph. 
We start by observing that the endpoints of $e$ and $e'$ must lie in the same
region of~$D(e,e')$. 
This fact was also used in earlier work by Aichholzer~et~al.~\cite{semisimple17} and
by Kyn{\v c}l~\cite{Kyncl2020}.  

\begin{lemma}\label{lem:all-out}
    The four vertices incident to~$e$ and~$e'$ belong to the same
    region of~$D(e,e')$.
\end{lemma}

\begin{proof}
    Assuming that the two edges cross at least two times, the
    drawing~$D(e,e')$ has at least two regions.
    Otherwise, the statement is trivial.
    If the four vertices do not belong to the same region of~$D(e,e')$, 
    then there is a vertex~$u$ of~$e$ and a vertex~$v$
    of~$e'$ that belong to different regions. 
    Now consider the edge~$uv$ in the drawing~$D$ of the complete graph. 
    This edge has ends in different regions of~$D(e,e')$, 
    whence it has a crossing with
    either~$e$ or~$e'$. This, however, makes a crossing in
    the star of~$u$ or~$v$. This contradicts the assumption that~$D$ is
    a star-simple drawing.
\end{proof}

Lemma~\ref{lem:all-out} implies that the deadlock
configurations as shown in \figurename~\ref{fig:deadlocks}
do not occur in star-simple drawings of complete graphs.
Formally, a \emph{deadlock} is a pair~$e,e'$ of edges such that
not all incident vertices lie in the same region of the drawing~$D(e,e')$.

Now suppose that~$D$ is a star-simple drawing of a complete graph with
no empty lens.  In this case we can argue that~$e$ and~$e'$ do not
form a configuration as the black edge~$e$ and the red edge~$e'$ in
\figurename~\ref{fig:spiral}. 
Indeed, that configuration has an
interior lens~$L$ and by assumption this lens is non-empty, i.e.,~$L$
contains a vertex~$x$. 
Let~$e$ and~$e'$ be the black and the red edge in
\figurename~\ref{fig:spiral}, respectively, and let~$u$ be a vertex
 of~$e$. The edge~$xu$ (the green edge in the figure) has no crossing with~$e$, hence
it follows the "tunnel" of the black edge. This yields a deadlock configuration of
the edges~$xu$ and~$e'$. Note that if in
\figurename~\ref{fig:spiral} instead of drawing the green edge $xu$ we connect~$x$ with an edge~$f$ to one of the vertices of the red edge~$e'$ such that~$f$ and the red edge have no crossing, then~$f$ and the black edge~$e$
form a deadlock.

\begin{figure}[bht]%
	\captionsetup[subfigure]{justification=centering}
	\hbox{}\hfill\begin{minipage}[b]{.63\linewidth}
		\centering
		\begin{tikzpicture}[scale=0.64]
  \coordinate (S) at (0,0);
  \coordinate (L) at ($(S)+(-0.5,-3.5)$);
  
  \fill[color=gray!20!white] ($(S)+(-2,0)$) arc (180:360:2);
  \fill[color=gray!20!white] (S) arc (0:190:1);
  \node at ($(S)+(0,-2)$) [anchor=north] {$e$};
  \fill ($(S)+(0,-0.5)$) circle (0.1);
  \draw[thick] ($(S)+(0,-0.5)$) -- (S);
  \draw[thick] (S) arc (0:180:1);
  \draw[thick] ($(S)+(-2,0)$) arc (180:360:2);
  \draw[thick] ($(S)+(2,0)$) -- ($(S)+(2,0.5)$);
  \fill ($(S)+(2,0.5)$) circle (0.1);
  \draw[color=red, thick] ($(S)+(-1,0)$) -- (2.5,0);
  \node[color=red] at ($(S)+(1,0)$) [anchor=south] {$e'$};
  \fill ($(S)+(-1,0)$) circle (0.1);
  \fill ($(S)+(2.5,0)$) circle (0.1);

  \coordinate (A) at ($(S)+(6,0)$);
  \fill[color=gray!20!white] ($(A)+(-2,0)$) arc (180:360:2);
  \node at ($(A)+(0,-2)$) [anchor=north] {$e$};
  \fill ($(A)+(0,-0.5)$) circle (0.1);
  \draw[thick] ($(A)+(0,-0.5)$) -- (A);
  \draw[thick] (A) arc (0:180:1);
  \draw[thick] ($(A)+(-2,0)$) arc (180:360:2);
  \draw[thick] ($(A)+(2,0)$) -- ($(A)+(2,0.5)$);
  \fill ($(A)+(2,0.5)$) circle (0.1);
  \fill ($(A)+(-1,0.5)$) circle (0.1);
  \draw[color=red, thick] ($(A)+(-2.5,0)$) -- ($(A)+(2.5,0)$);
  \node[color=red] at ($(A)+(1,0)$) [anchor=south] {$e'$};
  \fill ($(A)+(-2.5,0)$) circle (0.1);
  \fill ($(A)+(2.5,0)$) circle (0.1);
\end{tikzpicture}\vspace{9pt}
		\subcaption{deadlocks}
		\label{fig:deadlocks}
	\end{minipage}\hfill
	\begin{minipage}[b]{.37\linewidth}
		\centering
		\begin{tikzpicture}[scale=0.64]
  \coordinate (S) at (0,0);
  \coordinate (L) at ($(S)+(-0.5,-3.5)$);

  \coordinate (C) at ($(A)+(6.5,0)$);
  \fill[color=green!10!white] ($(C)+0.5*(-3,0)$) arc (180:360:2);
  \fill[color=green!70!black] ($(C)+0.5*(0.9,-0.4)$) circle (0.1);
  \node[color=green!70!black] at ($(C)+0.5*(0.9,-1)$)  [anchor=north] {$x$};
  \draw[color=green!70!black,thick] ($(C)+0.5*(1,0)$) arc (360:340:0.5);
  \draw[color=green!70!black,thick] ($(C)+0.5*(1,0)$) arc (0:180:1);
  \draw[color=green!70!black,thick] ($(C)+0.5*(-3,0)$) arc (180:360:2);
  \draw[color=green!70!black,thick] ($(C)+0.5*(5,0)$) arc (0:75:0.9);
  \fill ($(C)+0.5*(3.65,1.75)$) circle (0.1);
  \node at ($(C)+0.5*(3.65,1.75)$)  [anchor=south] {$u$};
  \fill ($(C)+0.5*(5.75,1.8)$) circle (0.1);
  \node at ($(C)+0.5*(1,-5)$)  [anchor=north] {$e$};
  \draw[thick] (C) arc (180:360:0.5);
  \draw[thick] (C) arc (0:180:0.5);
  \draw[thick] ($(C)+0.5*(2,0)$) arc (0:180:1.5);
  \draw[thick] ($(C)+0.5*(-2,0)$) arc (180:360:1.5);
  \draw[thick] ($(C)+0.5*(-4,0)$) arc (180:360:2.5);
  \draw[thick] ($(C)+0.5*(4,0)$) arc (0:20:2.5);
  \draw[thick] ($(C)+0.5*(6,0)$) arc (0:15:3.5);
  \draw[thick] (C) arc (180:360:0.5);
  \draw[color=red,thick] ($(C)+0.5*(-5,0)$) -- ($(C)+0.5*(7,0)$);
  \node[color=red] at ($(C)+0.5*(3,0)$) [anchor=south] {$e'$};
  \fill ($(C)+0.5*(-5,0)$) circle (0.1);
  \fill ($(C)+0.5*(7,0)$) circle (0.1);
  \fill ($(C)+0.5*(-1,0.5)$) circle (0.1);
\end{tikzpicture}
		\subcaption{spiral}
		\label{fig:spiral}
	\end{minipage}\hfill\hbox{}
	\caption{Constructions to achieve an unbounded number of crossings.}
	\label{fig:deadlock-spiral}
\end{figure}
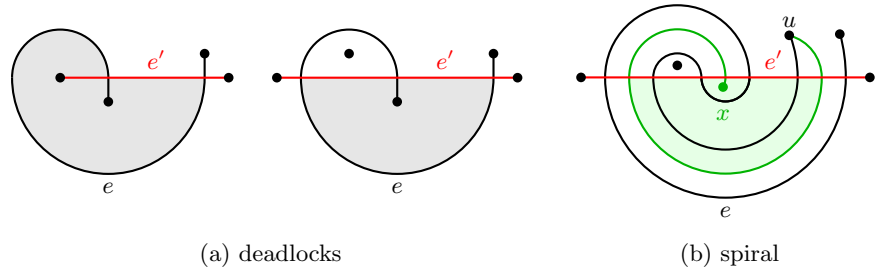

We use this intuition to formally define a spiral.  Two edges~$e,e'$
form a \emph{spiral} if they form a lens~$L$ such that if we place a
vertex~$x$ in~$L$ and draw a curve~$\gamma$ connecting~$x$ to a
vertex~$u$ of~$e$ so that~$\gamma$ does not cross~$e$, then~$\gamma$
and~$e'$ form a deadlock.
The discussion above proves the following lemma:
\begin{lemma}\label{lem:no-spiral}
 A star-simple drawing of a complete graph with
no empty lens has no pair $e$, $e'$ of edges that form a spiral.
 \end{lemma}  

\section{Crossings of pairs of edges}

In this section we derive an upper bound for the number of crossings
of two edges in a star-simple drawing of~$K_n$ with no
empty lens. 

\goodbreak
\begin{theorem}\label{thm:maxcross}
    Consider a star-simple drawing of~$K_n$ with no empty lens.
    If~$C(k)$ is the maximum number of crossings of a pair of edges that
    (a)~form no deadlock and no spiral and such that (b)~all lenses
    formed by the two edges can be hit by~$k$ points,
    then~$C(k) \leq \mathrm{e}\cdot k!$,
    where $\mathrm{e}\approx 2.718$ is Euler's number. 
\end{theorem}

{\renewcommand\endproof{\oldendproof}
\begin{proof}
    Due to Lemma~\ref{lem:all-out} we can assume that all four vertices of~$e$
    and~$e'$ are on the outer face of the drawing~$D(e,e')$.  We think of~$e'$ as being drawn red and horizontally and of~$e$ as being a black
    meander edge.  Let~$p_1,\ldots,p_k$ be points hitting all the lenses
    of the drawing~$D(e,e')$. Let~$u$ be one of the endpoints of~$e$. For
    each~$i=1,\ldots,k$ we draw an edge~$e_i$ connecting~$p_i$ to~$u$ such
    that~$e_i$ has no crossing with~$e$ and, subject to this, the number of
    crossings with~$e'$ is
    minimized. \figurename~\ref{fig:bigexample_reduced1}
    shows an example.

    Note that we do not claim that all these edges~$e_1,\ldots,e_k$
    together with~$e$ and~$e'$ can be extended to a star-simple
    drawing of a complete graph. Therefore, we cannot use
    Lemma~\ref{lem:no-spiral} directly but state the assumption~(a)
    instead.

    \begin{figure}[htbp]
        \centering
        \begin{tikzpicture}[scale=0.23]
	\coordinate (S) at (0,0);
	
\node[color=green!70!black] (L) at ($(S)+(24,-0.8)$) [anchor=north]{$p_i$};
\fill[color=green!70!black] ($(S)+(24,-0.5)$) circle (0.2);
\draw[color=green!70!black,thick] ($(S)+(24,-0.5)$) -- ($(S)+(24,0)$);
\draw[color=green!70!black,thick] ($(S)+(24,0)$) arc (0:180:4 and 2);
\draw[color=green!70!black,thick] ($(S)+(16,0)$) arc (360:180:4);
\draw[color=green!70!black,thick] ($(S)+(8,0)$) arc (0:180:2);
\draw[color=green!70!black,thick] ($(S)+(8,0)$) arc (0:180:2);
\draw[color=green!70!black,thick] ($(S)+(4,0)$) arc (180:353:12.55 and 9.5);
\node[color=green!70!black] at ($(S)+(23.8,-8)$) [anchor=north]{$e_i$};

  \draw[thick] ($(S)+(1,0)$) arc (180:360:1);
  \draw[thick] ($(S)+(3,0)$) arc (180:0:3);
  \draw[thick] ($(S)+(5,0)$) arc (180:0:1);
  \draw[thick] ($(S)+(5,0)$) arc (180:360:11 and 8);
  \draw[thick] ($(S)+(7,0)$) arc (180:360:5);
  \draw[thick] ($(S)+(9,0)$) arc (180:360:3);
  \draw[thick] ($(S)+(11,0)$) arc (180:360:1);
  \draw[thick] ($(S)+(11,0)$) arc (180:0:9 and 7);
  \node at ($(S)+(20,5.5)$) [anchor=south]{$e$};
  \draw[thick] ($(S)+(13,0)$) arc (180:0:7 and 5);
  \draw[thick] ($(S)+(15,0)$) arc (180:0:5 and 3);
  \draw[thick] ($(S)+(17,0)$) arc (180:0:1);
  \draw[thick] ($(S)+(19,0)$) arc (180:360:1);
  \draw[thick] ($(S)+(21,0)$) arc (180:0:1);
  \draw[thick] ($(S)+(23,0)$) arc (180:360:1);
  \draw[thick] ($(S)+(1,0)$) -- ($(S)+(1,1)$);
  \draw[thick] ($(S)+(29,0)$) -- ($(S)+(29,-1)$);

  \fill ($(S)+(2,-0.5)$) circle (0.2);
  \fill ($(S)+(6,0.5)$) circle (0.2);
  \fill ($(S)+(18,0.5)$) circle (0.2);
  \fill ($(S)+(22,0.5)$) circle (0.2);
  \fill ($(S)+(12,-0.5)$) circle (0.2);
  \fill ($(S)+(20,-0.5)$) circle (0.2);

  \fill ($(S)+(1,1)$) circle (0.2);
  \fill ($(S)+(29,-1)$) circle (0.2);
  \node (V) at ($(S)+(29,-1)$) [anchor=west]{$u$};

  \draw[color=red,thick] ($(S)-(1,0)$) -- ($(S)+(31,0)$);
  \node[color=red] at ($(S)+(10,0)$) [anchor=south]{$e'$};
  \fill ($(S)-(1,0)$) circle (0.2);
  \fill ($(S)+(31,0)$) circle (0.2);

\end{tikzpicture}
        \caption{The drawing~$D(e,e')$ and an edge~$e_i$ connecting~$p_i$ to~$u$.} 
        \label{fig:bigexample_reduced1}
    \end{figure}
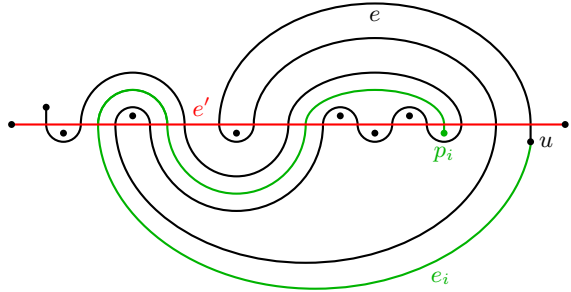
    
    \noindent We claim the following three properties:
    
    \begin{enumerate}[leftmargin=*,label={(P\arabic*)}]\setlength{\itemsep}{\labelsep}
        \item\label{prop:1} The edges $e_i$ and~$e'$ form no deadlock and no spiral.
        \item\label{prop:2} All 
          lenses of~$e_i$ and~$e'$ are hit by the~$k-1$
          points~$p_1,\ldots,p_{i-1},p_{i+1},\ldots,p_k.$
        \item\label{prop:3} Between any two crossings of~$e$ and~$e'$ from
        left to right, i.e., in the order along~$e'$, there is at least one
        crossing of~$e'$ with one of the edges~$e_i$.
    \end{enumerate}
    
    Before proving the properties, we show that they imply the
    statement of the theorem by induction on $k$. The base case
    $1=C(0) \leq \mathrm{e}\cdot 0! = \mathrm{e}$ is
    obvious. From
    \ref{prop:1}~and \ref{prop:2}~we see that the number~$X_i$ of
    crossings of~$e_i$ and~$e'$ is upper bounded by~$C(k-1)$. From
    \ref{prop:3}~we obtain that~$C(k) \leq 1+ \sum_i X_i$. Combining
    these we get

    \noindent\begin{minipage}[b]{.9\linewidth}
      \[
        C(k) \leq k\cdot C(k-1) +1 \leq k!\cdot
        \sum_{s=0}^{k}{\frac{1}{s!}} \leq k!\cdot \mathrm{e}.
      \]
    \end{minipage}
    \begin{minipage}[b]{.09\linewidth}
      \hfill\qed
    \end{minipage}
\end{proof}
}

For the proof of the three claims we need some notation.
Let~$\xi_1,\xi_2,\ldots,\xi_N$ be the crossings of~$e$ and~$e'$
indexed according to the left to right order along the horizontal
edge~$e'$.  Let~$g_i$ and~$h_i$ be the pieces of~$e'$ and~$e$,
respectively, between crossings~$\xi_i$ and~$\xi_{i+1}$. The bounded
region enclosed by~$g_i \cup h_i$ is the \emph{bag}~$B_i$ and~$g_i$ is
the \emph{gap} of the bag. In the drawing $D(e,e')$ the bags~$B_i$
where~$h_i$ is a crossing free piece of~$e$ are exactly the
inclusion-wise minimal lenses formed by~$e$ and~$e'$.  From now on when
referring to a \emph{lens} we always mean such a minimal lens. Indeed
if there is no empty minimal lens, then there is no empty lens.
The following observation is crucial.

\begin{obs}\label{obs:crucial}
    For two bags~$B_i$ and~$B_j$ the open interiors are either
    disjoint or one is contained in the other.
\end{obs}
\begin{proof}
  Every bag is bounded by a closed Jordan curve, and the boundaries of
  two distinct bags do not cross (at most they may touch at a single
  point that is one of $\xi_1,\xi_2,\ldots,\xi_N$).
\end{proof}

Observation~\ref{obs:crucial} implies that the containment order on the bags is a
downwards branching forest. The minimal elements in the containment
order are the lenses. Consider a lens~$L$ and the point~$p_i$
inside~$L$.  Since the vertex~$u$ of~$e$ is in the outer face
of~$D(e,e')$, the edge~$e_i$ has to leave each bag that contains~$L$.
Furthermore,  by definition~$e_i$ does not cross~$e$ and therefore it has to leave a bag~$B$ containing~$L$ through the
gap~$g$ of~$B$.
We now reformulate and prove the third claim (P3).

\begin{enumerate}[leftmargin=*,label={(P3')}]
    \item\label{prop:3p} For each pair~$\xi_i,\xi_{i+1}$ of consecutive crossings on~$e'$
    there is a lens~$L$ and a point~$p_j\in L$ such that~$e_j$
    crosses~$e'$ between~$\xi_i$ and~$\xi_{i+1}$.
\end{enumerate}

{\renewcommand\proofname{Proof sketch for \ref{prop:3p}}%
  \begin{proof}
    The pair~$\xi_i,\xi_{i+1}$ is associated with the bag~$B_i$. In the
    containment order of bags a minimal bag below~$B_i$ is a lens,
    let~$L$ be any of the minimal elements below~$B_i$. 
    By assumption, $L$ contains a point $p_j$. 
    Since~$L \subseteq B_i$, we have that also $p_j \in B_i$. 
    Thus, 
    it follows that~$e_j$ has a crossing with
    the gap~$g_i$, i.e.,~$e_j$ has a crossing with~$e'$ between~$\xi_i$
    and~$\xi_{i+1}$.
\end{proof}}

{\renewcommand\proofname{Proof sketch for \ref{prop:1}}%
\begin{proof}
    We have to show that~$e_i$ and~$e'$ form no deadlock
    and no spiral. The minimality condition in the definition of~$e_i$
    implies that
    if~$L=B_{i_1}\subset B_{i_2}\subset \ldots \subset B_{i_t}$ is the
    maximal chain of bags with minimal element~$L$, then~$e_i$
    crosses the gaps of these bags in the given order and has no
    further crossings with~$e'$. If~$\gamma$ is a curve from~$L$ to~$u$ that
    avoids~$e$, then in the ordered sequence of gaps crossed by~$\gamma$
    we find a subsequence that is identical to the ordered sequence of
    gaps crossed by~$e_i$.  Since~$e$ and~$e'$ form no spiral, there is
    such a curve~$\gamma$ that forms no deadlock with~$e'$. Therefore,
    $e_i$ forms no deadlock with~$e'$, either.
    
    Now assume that~$e_i$ and~$e'$ form a spiral. Let~$B$ be the largest
    bag containing~$p_i$. Think of~$B$ as a drawing of~$e_i$ with a
    broad pen, which may also have some extra branches that have no
    correspondence in~$e_i$, see \figurename~\ref{fig:no-spiral}.  The
    formalization of this picture is that for every bag~$\beta$ formed
    by~$e_i$ with~$e'$ there is a bag~$B(\beta)$ formed by~$e$ and~$e'$
    with~$B(\beta) \subset \beta$. Now, if there is a lens~$\lambda$
    formed by~$e_i$ with~$e'$ such that every~$e_i$-avoiding\footnote{that is, disjoint from $e_i$ except for possibly a shared endpoint} curve
    to~$u$ is a deadlock with~$e'$, then there is a lens~$L(\lambda)$
    formed by~$e$ and~$e'$ with~$L(\lambda) \subset \lambda$ such that
    every~$e$-avoiding curve from~$L(\lambda)$ to~$u$ is also~$B$-avoiding and
    hence~$e_i$-avoiding.  Thus, 
    every such curve has a
    deadlock with~$e'$, whence~$e$ and~$e'$ form a spiral, contradiction.
\end{proof}}

{\renewcommand\proofname{Proof sketch for \ref{prop:2}}%
\begin{proof}
  We 
  know by \ref{prop:1} that~$e_i$ and~$e'$ form no deadlock. Therefore, by
    Lemma~\ref{lem:all-out}, the vertices of~$e_i$ and~$e'$ belong to
    the same region of~$D(e_i,e')$.  All crossings of~$e_i$ with~$e'$
    correspond to bags of~$e$ and~$e'$, therefore the vertices of~$e$
    and~$e'$ are in the outer face of~$D(e_i,e')$.  Together this shows
    that~$p_i$ is also in the outer face of~$D(e_i,e')$. Since every
    lens of~$D(e_i,e')$ contains a lens of~$D(e,e')$, it also contains
    one of the points hitting all lenses of~$D(e,e')$.  Hence, all
    lenses of~$D(e_i,e')$ are hit by the~$k-1$
    points~$p_1,\ldots,p_{i-1},p_{i+1},\ldots,p_k$.
\end{proof}}

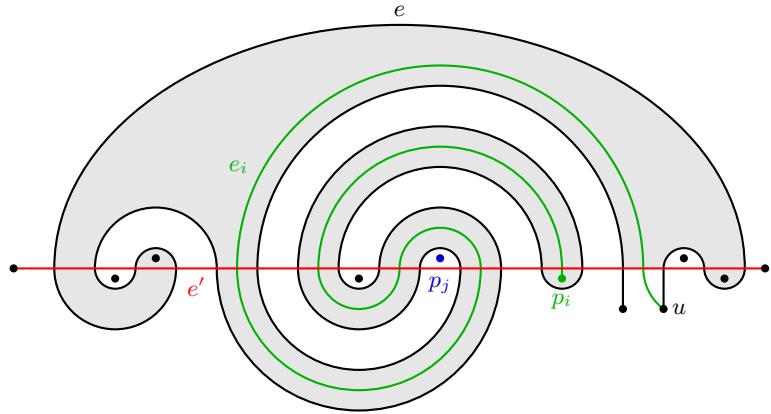
\begin{figure}[thb]
    \centering
    \begin{tikzpicture}[scale=0.27]
	\coordinate (S) at (0,0);
  \draw[fill=gray!20!white,thick] ($(S)+(1,0)$) arc (180:360:3);
  \draw[fill=gray!20!white,thick] ($(S)+(1,0)$) arc (180:0:17 and 12);
  \node at ($(S)+(18,12)$) [anchor=south]{$e$};
  \draw[fill=white,thick] ($(S)+(3,0)$) arc (180:360:1);
  \draw[fill=white,thick] ($(S)+(3,0)$) arc (180:0:3);
  \draw[fill=gray!20!white,thick] ($(S)+(5,0)$) arc (180:0:1);
  \draw[fill=gray!20!white,thick] ($(S)+(9,0)$) arc (180:360:7);
  \draw[fill=white,thick] ($(S)+(11,0)$) arc (180:360:5);
  \draw[fill=white,thick] ($(S)+(11,0)$) arc (180:0:9);
  \draw[fill=gray!20!white,thick] ($(S)+(13,0)$) arc (180:360:3);
  \draw[fill=gray!20!white,thick] ($(S)+(13,0)$) arc (180:0:7);
  \draw[fill=white,thick] ($(S)+(15,0)$) arc (180:360:1);
  \draw[fill=white,thick] ($(S)+(15,0)$) arc (180:0:5);
  \draw[fill=gray!20!white,thick] ($(S)+(17,0)$) arc (180:0:3);
  \draw[fill=white,thick] ($(S)+(19,0)$) arc (180:0:1);
  \draw[fill=gray!20!white,thick] ($(S)+(25,0)$) arc (180:360:1);
  \draw[fill=white,thick] ($(S)+(31,0)$) arc (180:0:1);
  \draw[fill=gray!20!white,thick] ($(S)+(33,0)$) arc (180:360:1);
  \draw[thick] ($(S)+(29,0)$) -- ($(S)+(29,-2)$);
  \draw[thick] ($(S)+(31,0)$) -- ($(S)+(31,-2)$);

	\node[color=green!70!black] at ($(S)+(26,-0.8)$) [anchor=north]{$p_i$};
  \fill[color=green!70!black] ($(S)+(26,-0.5)$) circle (0.2);
  \draw[color=green!70!black,thick] ($(S)+(26,-0.5)$) -- ($(S)+(26,0)$);
  \draw[color=green!70!black,thick] ($(S)+(26,0)$) arc (0:180:6);
  \draw[color=green!70!black,thick] ($(S)+(14,0)$) arc (180:360:2);
  \draw[color=green!70!black,thick] ($(S)+(18,0)$) arc (180:0:2);
  \draw[color=green!70!black,thick] ($(S)+(22,0)$) arc (360:180:6);
  \draw[color=green!70!black,thick] ($(S)+(10,0)$) arc (180:0:10);
  \draw[color=green!70!black,thick] ($(S)+(30,0)$) arc (180:230:2.5);
  \node[color=green!70!black] at ($(S)+(11,5)$) [anchor=east]{$e_i$};

  \fill ($(S)+(4,-0.5)$) circle (0.2);
  \fill ($(S)+(6,0.5)$) circle (0.2);
  \fill ($(S)+(16,-0.5)$) circle (0.2);
  \fill[color=blue!80!black] ($(S)+(20,0.5)$) circle (0.2);
  \node[color=blue!80!black] at ($(S)+(20,0)$) [anchor=north]{$p_j$};
  \fill ($(S)+(32,0.5)$) circle (0.2);
  \fill ($(S)+(34,-0.5)$) circle (0.2);

  \fill ($(S)+(29,-2)$) circle (0.2);
  \fill ($(S)+(31,-2)$) circle (0.2);
  \node at ($(S)+(31,-2)$) [anchor=west]{$u$};

  \draw[color=red,thick] ($(S)-(1,0)$) -- ($(S)+(36,0)$);
  \node[color=red] at ($(S)+(8,0)$) [anchor=north]{$e'$};
  \fill ($(S)-(1,0)$) circle (0.2);
  \fill ($(S)+(36,0)$) circle (0.2);

\end{tikzpicture}
    \caption{An edge~$e_i$(green) that forms a spiral with~$e'$.
        The bag~$B$ in gray and the lens~$L(\lambda)$ marked with the
        vertex~$p_j$(blue).} 
    \label{fig:no-spiral}
  \end{figure}

\section{Crossings in complete drawings}

Accounting for the four endpoints of the two crossing edges we have
$k\le n-4$ in Theorem~\ref{thm:maxcross}. Therefore, we obtain that
the number of crossings of a pair of edges in a star-simple drawing
of~$K_n$ without empty lens is upper bounded
by~$3(n-4)!$. This directly implies that the drawing of~$K_n$
has at most~$n!$ crossings. 
This is the first finite upper bound on the number of crossings in star-simple drawings of the complete graph $K_n$. 
 We know drawings of~$K_n$ in this drawing
mode that have an exponential number of crossings. 
Thus, it would be interesting
to reduce the huge gap between the upper and the lower bound.
Specifically, can a star-simple drawing of~$K_n$ have two edges with more than~$2^{n-4}+2^{n-6}$ crossings?

%
%
%
\bibliographystyle{splncs04}
\bibliography{gd20-star-simple}
\end{document}